\documentclass[12pt, reqno]{amsart}
\usepackage{amsmath, amssymb}
\usepackage{graphicx, color, wrapfig, framed}
\usepackage[height=22cm, width=15.5cm, hmarginratio={1:1}]{geometry} 
\usepackage{amsthm, array}
\usepackage{graphicx,xcolor}

\usepackage[hypertexnames=false,hyperfootnotes=false,colorlinks=true,linkcolor=blue,%
citecolor=purple,filecolor=magenta,urlcolor=cyan,unicode,linktocpage=true,pagebackref=false]{hyperref}
\usepackage{nameref,zref-xr}

\numberwithin{equation}{section}

\newcommand{\beq}{\begin{equation}}
	\newcommand{\eeq}{\end{equation}}
\newcommand{\bal}{\begin{align}}
	\newcommand{\eal}{\end{align}}
	
\def\thin{\hskip 1 pt} 

\newcommand{\e}{\epsilon}
\newcommand{\p}{\partial}
\newcommand{\nn}{\nonumber}

\newcommand{\CC}{\mathbb{C}}
\newcommand{\QQ}{\mathbb{Q}}

\newcommand{\Mgnbar}{\overline{\mathcal{M}}_{g,n}}

\newcommand{\mbE}{\mathbb{E}}
\newcommand{\mbC}{\mathbb{C}}
\newcommand{\hcA}{\widehat{\mathcal{A}}}
\newcommand{\hcF}{\widehat{\mathcal{F}}}
\newcommand{\oh}{\overline{h}}
\newcommand{\cS}{\mathcal{S}}

\newcommand{\tr}{\widetilde{r}}
\newcommand{\tg}{\widetilde{g}}
\newcommand{\tc}{\widetilde{c}}
\newcommand{\ovr}{\overline{r}}
\newcommand{\mcO}{\mathcal{O}}
\newcommand{\cP}{\mathcal{P}}
\newcommand{\mbZ}{\mathbb{Z}}
\newcommand{\Mi}{\mathrm{Mi}}
\newcommand{\tH}{\widetilde{H}}
\newcommand{\tomega}{\widetilde{\omega}}
\newcommand{\triv}{\mathrm{triv}}
\newcommand{\const}{\mathrm{const}}
\newcommand{\tu}{\widetilde{u}}
\newcommand{\oM}{\overline{\mathcal{M}}}
\def\d{\partial}
\newcommand{\Coef}{\mathrm{Coef}}
\newcommand{\lb}{\left(}
\newcommand{\rb}{\right)}
\newcommand{\DR}{\mathrm{DR}}
\newcommand{\mbQ}{\mathbb{Q}}
\newcommand{\CP}{\mathbb{CP}}
\newcommand{\og}{\overline{g}}
\newcommand{\tsigma}{\widetilde{\sigma}}
\newcommand{\eps}{\epsilon}
\newcommand{\od}{\overline{d}}

\newcommand{\cA}{\mathcal{A}}

\newtheorem{thm}{Theorem}[section]

\newtheorem{lem}[thm]{Lemma}
\newtheorem{Def}[thm]{Definition}
\newtheorem{cor}[thm]{Corollary}
\newtheorem{rmk}[thm]{Remark}
\newtheorem{conj}[thm]{Conjecture}
\newtheorem{exa}[thm]{Example}

\newtheorem{manualtheoreminner}{Theorem}
\newenvironment{manualtheorem}[1]{%
	\renewcommand\themanualtheoreminner{#1}%
	\manualtheoreminner
}{\endmanualtheoreminner}

\title[Bihamiltonian tests for integrable systems associated to F-CohFTs]{Bihamiltonian tests for integrable systems associated to rank-$1$ F-CohFTs}
\author{Alexandr Buryak, Jianghao Xu, Di Yang}
\address{A. Buryak, Faculty of Mathematics, National Research University Higher School of Economics, Usacheva str. 6, Moscow, 119048, Russian Federation;
Skolkovo Institute of Science and Technology, Bolshoy Boulevard 30, bld. 1, Moscow, 121205, Russian Federation}
\email{aburyak@hse.ru}
\address{Jianghao Xu, School of Mathematical Sciences, University of Science and Technology of China, 230026 Hefei, P.R.~China}
\email{xjh\_020403@mail.ustc.edu.cn}
\address{Di Yang, School of Mathematical Sciences, University of Science and Technology of China, 230026 Hefei, P.R.~China}
\email{diyang@ustc.edu.cn}
\begin{document}

\begin{abstract}
Double ramification (DR) hierarchies associated to rank-$1$ F-CohFTs are important integrable perturbations of the Riemann--Hopf hierarchy. In this paper, we perform bihamiltonian tests for these DR hierarchies, and conjecture that the ones that are bihamiltonian form a $2$-parameter family. Remarkably, our computations suggest that there is a $1$-parameter subfamily of the rank-$1$ F-CohFTs, where the corresponding DR hierarchy is conjecturally Miura equivalent to the Camassa--Holm hierarchy. We also prove a conjecture regarding bihamiltonian Hodge hierarchies. Finally, we systematically study Miura invariants, and for another $1$-parameter subfamily propose a conjectural relation to the Degasperis--Procesi hierarchy.
\end{abstract}

\date{\today}

\maketitle

\tableofcontents

\section{Introduction}

The {\em Korteweg--de Vries (KdV) equation} 
\beq\label{KdVeq}
u_t = u u_x +\frac{\e^2}{12} u_{xxx}
\eeq
and the {\em Camassa--Holm (CH) equation}
\beq\label{CHequation}
v_t-\e^2 v_{xxt}= 
v v_x
-\e^2 \Bigl(\thin \frac23 v_x v_{xx}+\frac13 v v_{xxx}\Bigr) 
\eeq
are two important nonlinear partial differential equations (PDEs). 
One crucial feature of these two PDEs is that they have bihamiltonian structures~\cite{CH93, FF81, Magri}. 
The bihamiltonian structure of the KdV equation reads
\begin{equation*}
u_t = \cP_1 \Bigl(\frac{\delta H_1}{\delta u}\Bigr) = \frac23 \cP_2 \Bigl(\frac{\delta H_0}{\delta u}\Bigr),
\end{equation*}
where 
\beq\label{KdVbihamiOp}
\cP_1 = \p_x,\quad \cP_2 = u \p_x +\frac12 u_x + \frac{\e^2}8 \p_x^3
\eeq
are compatible Poisson operators and 
$H_1 = \int\bigl(\frac{1}{6}u^3-\frac{\e^2}{24}u_x^2\bigr)dx$, 
$H_0 =\int\frac{1}{2} u^2dx$
are Hamiltonians. The bihamiltonian structure of the CH equation reads
\begin{equation*}
m_t = \cP_1 \Bigl(\frac{\delta H_1}{\delta m}\Bigr) = \frac23 \cP_2 \Bigl(\frac{\delta H_0}{\delta m}\Bigr),
\end{equation*}
where $m:=v-\e^2 v_{xx}$,   
\beq\label{CHbihamiOp}
\cP_1 = \p_x-\e^2 \p_x^3,\quad \cP_2 = m \p_x +\frac12 m_x
\eeq
are compatible Poisson operators and 
$H_1 = \int(\frac{1}{6}v^3-\frac{\e^2}{12}v^2 v_{xx})dx$, 
$H_0 =\int(\frac{1}{2} v^2 -\frac{\e^2}2 v v_{xx})dx$
are Hamiltonians.
Here and below $\frac{\delta}{\delta u}$ denotes the variational derivative,~i.e.,
\beq\label{VariationalDerivative}
\frac{\delta(\int f dx)}{\delta u}:=\sum_{k=0}^{\infty} (-\p_x)^k \Bigl(\frac{\p f}{\p u_{k}}\Bigr), \quad u_{k}:=\p^{k}_x(u).
\eeq
Moreover, each of \eqref{KdVeq}, \eqref{CHequation}
can \cite{CH93, Dickey, DYZ, FF81, GGKM, Magri} be extended to a hierarchy of pairwise commuting PDEs, called the {\em KdV hierarchy} and the {\em CH hierarchy}, respectively.

\medskip

The last 30 years have witnessed the discovery of deep relations between integrable hierarchies and the topology of $\Mgnbar$ --- the Deligne--Mumford moduli spaces of stable algebraic curves. These deep relations were first manifested in Witten's conjecture~\cite{Witten}, which was proved by Kontsevich \cite{Kont}. This result established a remarkable relationship between a particular $\tau$-function of the KdV hierarchy and the partition function of the intersection numbers of $\psi$-classes on $\Mgnbar$. In \cite{BPS, DubrovinGW13, DZnormalform}, for any given semisimple cohomological field theory (CohFT), the {\em Dubrovin--Zhang (DZ) integrable hierarchy} was defined, whose topological partition function is the total descendent potential of the CohFT. 
In~\cite{Buryak15}, another integrable hierarchy associated to a given CohFT
was constructed using the double ramification (DR) cycles on the moduli spaces of curves. This hierarchy is called the {\em DR hierarchy}. Both the DZ hierarchy and the DR hierarchy are now known to be Miura equivalent \cite{BDGR18,BDGR20,BGR19,BLS24}.

\medskip

Both the KdV hierarchy and the CH hierarchy have the same dispersionless part, which is the {\em Riemann--Hopf (RH) hierarchy}
\begin{equation}\label{eq:RH hierarchy}
u_{t_d}=\frac{u^d}{d!}u_x,\quad d\ge 1.
\end{equation}
In \cite{DLYZ16}, a class of integrable hierarchies, called {\em $\tau$-symmetric Hamiltonian deformations of the Riemann--Hopf hierarchy}, was studied. It was conjectured in~\cite{DLYZ16} that the {\em Hodge integrable hierarchy} \cite{BPS, DLYZ16} (for short the {\it Hodge hierarchy}), which is the DZ hierarchy associated to a rank-$1$ CohFT $\{c_{g,n}\}$ and contains an infinite family of parameters
given by an $R$-matrix $R(z)=\exp(\sum_{i\ge 1}r_{2i-1}z^{2i-1})$, $r_{2i-1}\in\mbC$, is a universal object for all $\tau$-symmetric Hamiltonian deformations of the Riemann hierarchy. This is the content of the {\em Hodge Universality Conjecture}, which is still an open problem; see also~\cite{BDGR20, BR25}.
Another important question is about when the Hodge hierarchy has a bihamiltonian structure. In~\cite{DLYZ16}, the bihamiltonian test for the Hodge hierarchy was made and it was conjectured that the Hodge hierarchy 
associated to a rank-$1$ CohFT $\{c_{g,n}\}$ has a DZ-type bihamiltonian structure if and only if 
$c_{g,n}$ is given by 
\begin{equation}\label{eq:CohFT for Volterra}
c_{g,n}^{\rm special}=\Lambda(2q)^2\Lambda(-q)\in H^*(\Mgnbar),\quad q\in\mbC,
\end{equation}
where $\Lambda(z):=1+z\lambda_1+\ldots+z^g\lambda_g$, $\lambda_i:=c_i(\mbE)\in H^{2i}(\Mgnbar)$, and $\mbE$ is the Hodge bundle on $\Mgnbar$.  
The particular CohFT $c_{g,n}^{\rm special}$ corresponds to the specialization
\beq\label{specializationofHodge}
r_{2i-1}=r_{2i-1}^{\rm special}=\frac{B_{2i}}{2i(2i-1)}(2^{2i}-1) q^{2i-1},\quad q\in\mbC,
\eeq
where $B_{k}$ denotes the $k$th Bernoulli number. The fact that the Hodge hierarchy under the specialization~\eqref{specializationofHodge} is bihamiltonian was proved~\cite{DLYZ20} 
as a consequence of the Hodge--GUE correspondence~\cite{DLYZ20, DY17}.
In this paper, we will prove the remaining part of the conjecture by using an argument from~\cite{YZ2023} and a result from \cite{BR25}.

\medskip

In~\cite{ALM} Arsie--Lorenzoni--Moro (ALM) studied integrable systems of conservation laws that are deformations of the RH hierarchy. ALM showed that there exists a unique Miura-type transformation that transforms such a deformation to a certain {\em normal form}. Moreover, ALM conjectured~\cite{ALM} that any such normal form is uniquely determined by the coefficients of the quasi-linear part of the first flow and that these coefficients can be arbitrary functions of one variable. Following the terminology from~\cite{AL18}, we call these coefficients {\em Miura invariants}. The conjecture from~\cite{ALM} will be called the {\em Arsie--Lorenzoni--Moro conjecture}.

\medskip

Recently, the DR hierarchy associated to an F-CohFT (which is a generalization of the notion of a CohFT) was constructed in~\cite{ABLR21}. Arsie--Buryak--Lorenzoni--Rossi also constructed~\cite{ABLR21} a family of F-CohFTs, which in the rank-$1$ case are parameterized by an $R$-matrix $R(z)=\exp(\sum_{i\ge 1}r_{i}z^i)$, $r_i\in\mbC$. The DR hierarchies associated to these rank-$1$ F-CohFTs were studied in details in~\cite{BR25}. Buryak--Rossi proved~\cite{BR25} that these DR hierarchies are in the normal form, have constant Miura invariants, which can be arbitrary. Also, the DR hierarchies associated to rank-$1$ F-CohFTs with different $R$-matrices have different Miura invariants. In this paper, we will do the bihamiltonian tests for the DR hierarchies under consideration. Our computations give enough evidence to conjecture that the DR hierarchies (associated to the rank-$1$ F-CohFTs with $R$-matrices $R(z)=\exp(\sum_{i\ge 1}r_{i}z^i)$) that are bihamiltonian form a $2$-parameter family. A $1$-parameter subfamily in this family corresponds to the CohFT~\eqref{eq:CohFT for Volterra}. Remarkably, our computations suggest that there is another $1$-parameter subfamily, where the corresponding DR hierarchy is Miura equivalent to the CH hierarchy. As far as we know, this is the first time that a connection of the CH hierarchy to the topology of~$\Mgnbar$ is found. 

\medskip

Note that the DZ hierarchy corresponding to an arbitrary F-CohFT can be constructed and both the DZ hierarchy and the DR hierarchy are Miura equivalent, see~\cite{BS24,BLS24}. 

\medskip

There is a simple procedure that allows, starting from the DR hierarchy corresponding to a rank-$1$ F-CohFT, to get a hierarchy of conservation laws with not necessarily constant Miura invariants. Namely, one can take a linear combination of the flows of the DR hierarchy and then make a Miura-type transformation after which this flow becomes $u_s=\p_x(\frac{u^2}{2}+O(\e))$. The Miura invariants of this flow are not necessarily constants. We will find an explicit formula for these Miura invariants. We will also show that one can obtain in this way integrable hierarchies of conservation laws with arbitrary Miura invariants that depend linearly on~$u$. This implies that if we assume that the ALM conjecture is true, then there exists an $R$-matrix such that the corresponding DR hierarchy is Miura equivalent to the CH hierarchy, which agrees with our conjecture from above. Remarkably, we also see that the ALM conjecture implies that there exists another $R$-matrix such that the associated DR hierarchy is Miura equivalent to the integrable hierarchy corresponding to the {\em Degasperis--Procesi (DP) equation}~\cite{DP}
$$
v_t-\e^2 v_{xxt}= 
v v_x
-\e^2 \Bigl(\,\frac{3}{4} v_x v_{xx}+\frac{1}{4}v v_{xxx}\Bigr),
$$
which, as far as we know, never appeared among the integrable systems related to the moduli spaces of curves $\Mgnbar$.  

\medskip

The paper is organized as follows. In Section~\ref{section:preliminaries about PDEs}, we review some terminologies and results about evolutionary PDEs of DZ type.  In Section~\ref{review}, we study integrable systems associated to rank-$1$ F-CohFTs. In Section~\ref{bihamiHodge}, we prove the conjecture from~\cite{DLYZ16} on bihamiltonian Hodge hierarchies. In Section~\ref{bihamitest}, we perform bihamiltonian tests for the DR hierarchies for rank-$1$ F-CohFTs. In Section~\ref{Miura-invariants-transf}, we systematically study the Miura invariants of the DR hierarchies associated to rank-$1$ F-CohFTs.

\medskip

\noindent{\bf Notation and conventions.}
\begin{itemize}
\item For a topological space $X$, we denote by $H^i(X)$ the cohomology groups with the coefficients in $\mbC$. Let $H^{\mathrm{e}}(X):=\bigoplus_{i\ge 0}H^{2i}(X)$. 

\smallskip

\item We will work with the moduli spaces $\Mgnbar$ of stable algebraic curves of genus $g$ with $n$ marked points, which are nonempty only when the condition $2g-2+n>0$ is satisfied. We will often omit mentioning this condition explicitly, and silently assume that it is satisfied when a moduli space is considered.

\smallskip

\item We denote by $\cP_n$ the set of partitions of $n\in\mbZ_{\ge 0}$. For $\lambda=(\lambda_1,\ldots,\lambda_l)\in\cP_n$, let $m_k(\lambda):=|\{1\le i\le l|\lambda_i=k\}|$.
\end{itemize}

\medskip

\noindent{\bf Acknowledgements.} The work of A.~B. is an output of a research project implemented as part of the Basic Research Program at the National Research University Higher School of Economics (HSE University). The work of J. Xu and D. Yang is supported by NSFC No.~12371254 and CAS YSBR-032.

\medskip

\section{Preliminaries on evolutionary PDEs of DZ type}\label{section:preliminaries about PDEs}

In this section, we review the necessary terminology and results about evolutionary PDEs of DZ type, 
referring to \cite{DLZ06, DZnormalform} (cf.~also~\cite{YZ2023}) for more details.

\medskip

\subsection{Evolutionary PDEs of DZ type and Miura-type transformations}

In this paper, we consider scalar ($1$+$1$)-dimensional evolutionary PDEs. Here ``scalar ($1$+$1$)-dimensional'' means that the number of unknown functions is one, and the unknown function has one space variable and one time variable. And by saying ``evolutionary" we mean that the PDE expresses the partial derivative with respect to the time variable as a function of the partial derivatives with respect to the space variable.

\medskip

Let $\mathcal{A}_u:=\mathcal{S}(u)[u_1,u_2,\ldots]$ be the differential polynomial ring of $u$, where $\mathcal{S}(u)$ is a suitable ring of functions of $u$. For instance, $\mathcal{S}(u)$ could be $\mathcal{O}(B)$, the ring of analytic functions on an open disk $B\subset\mbC$. When $u$ is taken as a function of $x$, we can identify~$u_m$ with $\p^m u/\p x^m$. 
Introduce an operator $\p_x$ in $\mathcal{A}_u$ by $\p_x:=\sum_{m\geq 0}u_{m+1}\p/\p u_{m}$. Define a gradation $\deg$ on $\mathcal{A}_u$ by assigning $\deg u_m:=m$ $(m\geq 1)$, and we use $\mathcal{A}_u^{[k]}$ to denote the set of elements in $\mathcal{A}_u$ that are graded homogeneous of degree $k$ with respect to $\deg$. Let us extend the gradation $\deg$ on $\mathcal{A}_u$ to a gradation on $\hcA_u:=\mathcal{A}_u[[\e]]$ by assigning $\deg\e:=-1$. We denote by $\hcA_u^{[k]}\subset\hcA_u$ the homogeneous component of degree~$k$. To a differential polynomial $P\in\hcA_u$ we assign the derivation
$$
D_P:=\sum_{n\ge 0}(\d_x^n P)\frac{\d}{\d u_n}\colon\hcA_u\to\hcA_u
$$
called the \emph{evolutionary operator}.

\medskip

By an {\em evolutionary PDE of hydrodynamic type}, we mean a PDE of the form
\beq\label{HydrodynamicPDE}
\frac{\p u}{\p t}=S(u) \frac{\p u}{\p x},\quad S(u)\in\mathcal{S}(u).
\eeq
For example, the RH hierarchy~\eqref{eq:RH hierarchy} consists of infinitely many evolutionary PDEs of hydrodynamic type. By an {\em evolutionary PDE of DZ type} we mean
\beq\label{DZnormalform}
u_t=\sum_{k\geq 0}\e^k S_k(u,u_1,\ldots,u_{k+1}),\quad S_k\in \mathcal{A}_u^{[k+1]}.
\eeq
If $S_0(u,u_1)=S(u)u_1$, then we call \eqref{DZnormalform} a {\em perturbation} of~\eqref{HydrodynamicPDE}, 
and call~\eqref{HydrodynamicPDE} the {\em dispersionless limit} of~\eqref{DZnormalform}. Note that if $S(u)\ne\mathrm{const}$, then any PDE 
\begin{equation*}
u_s=\sum_{k\geq 0}\e^k Q_k(u,u_1,\ldots,u_{k+1}),\quad Q_k\in \mathcal{A}_u^{[k+1]},
\end{equation*}
that commutes with~\eqref{DZnormalform} is uniquely determined~\cite{LZ06} by the differential polynomials~$Q_0$ and $S_0,S_1$, $\dots$.

\medskip

Denote by $\int(\cdot) dx\colon \mathcal{A}_u \to \mathcal{A}_u /(\p_x \mathcal{A}_u\oplus\mbC)$ the canonical projection, which can be termwise extended to a projection $\hcA_u \to\hcF_u:=\hcA_u/(\p_x \hcA_u\oplus\mbC[[\e]])$. We call elements of $\hcF_u$ {\em local functionals}. The variational derivative $\frac{\delta}{\delta u}$ defined by \eqref{VariationalDerivative} is well defined on~$\hcF_u$ and gives a linear map $\frac{\delta}{\delta u}\colon\hcF_u\to\hcA_u$. It is known that a local functional $\oh\in\hcF_u$ is equal to zero if and only if $\frac{\delta\oh}{\delta u}=0$. A local functional $\oh\in\hcF_u$ is called a \emph{conserved quantity} of an evolutionary PDE $u_t=P$, $P\in\hcA_u^{[1]}$, if $D_P(\oh)=0$. 

\medskip

Let $\cP$ be a nonzero operator of the form
\beq\label{PoissonOperator}
\cP=\sum_{k\geq 0}\e^k \cP^{[k]},\quad \cP^{[k]}=\sum_{j=0}^{k+1}A_{k,j}\p_x^j,\quad A_{k,j}\in\mathcal{A}_u^{[k+1-j]}.
\eeq
The bracket $\{\cdot,\cdot\}_{\cP}\colon\hcF_u\times\hcF_u\to\hcF_u$ associated to the operator $\cP$ is defined by
$$
\biggl\{\int f dx,\int g dx\biggr\}_{\cP} := \int \frac{\delta f}{\delta u} \cP\Bigl( \frac{\delta g}{\delta u} \Bigr) dx.
$$
By definition the bracket is bilinear. If $\{\,,\,\}_{\cP}$ is anti-symmetric and satisfies the Jacobi identity, then we call it a {\it Poisson bracket} 
and the operator~$\cP$ a {\it Poisson operator}. The part $\cP^{[0]}$ in \eqref{PoissonOperator} is called the dispersionless limit of~$\cP$. If $\cP$ is a Poisson operator, then~$\cP^{[0]}$ is also a Poisson operator. We call a nonzero Poisson operator like $\cP^{[0]}$ a {\em Poisson operator of hydrodynamic type}. In \cite{DN83}, a Poisson operator of hydrodynamic type associated to a contravariant flat pseudo-Riemannian metric was constructed. In our case with one dependent variable such $\cP^{[0]}$ simply has the form
$$
\cP^{[0]}=g(u)\p_x+\frac12 g'(u)u_1,
$$
where $g(u)\in\cS(u)$ is an arbitrary non-zero function, which is the contravariant metric associated to~$\cP^{[0]}$.

\medskip

A {\em Hamiltonian PDE of DZ type} is a PDE of the form~\eqref{DZnormalform} such that there exists a Poisson operator~$\cP$ and a local functional $\oh\in\hcF_u^{[0]}$ satisfying
\beq\label{hamiltonianPDE}
\frac{\p u}{\p t}=\cP\Bigl(\frac{\delta \oh}{\delta u}\Bigr).
\eeq
The local functional $\oh$ is called the {\em Hamiltonian} of \eqref{hamiltonianPDE}, and a differential polynomial $h\in\hcA_u^{[0]}$ satisfying $\oh=\int h dx$ is called the {\em Hamiltonian density}. A local functional $\ovr=\int r dx$, $r\in\hcA_u^{[0]}$, is called a {\em Casimir} for the Poisson operator~$\cP$~if
$$
\cP \Bigl(\frac{\delta\ovr}{\delta u}\Bigr)=0.
$$
An evolutionary PDE of DZ type is called {\em bihamiltonian} if it can be written as a Hamiltonian PDE in two ways,
\begin{gather}\label{eq:bihamiltonian PDE}
\frac{\p u}{\p t}=\cP_1\Bigl(\frac{\delta\oh_1}{\delta u}\Bigr)=\cP_2\Bigl(\frac{\delta\oh_2}{\delta u}\Bigr),
\end{gather}
where the Poisson operators $\cP_1, \cP_2$ are {\it compatible}, i.e., an arbitrary linear combination of $\cP_1,\cP_2$ is again a Poisson operator, 
and $\cP_1\ne 0$ and $\cP_2\ne\lambda \cP_1$ for any $\lambda\in\mbC$.

\medskip

We will call a transformation of the form
\beq\label{Miuratypetrans}
u\mapsto \tilde{u}=\sum_{k\geq 0}\e^k M^{[k]}(u,u_1,\ldots,u_{k})\in\hcA^{[0]}_u,\quad \frac{\p M^{[0]}(u)}{\p u}\ne 0,
\eeq
a {\em Miura-type transformation}. The class of evolutionary PDEs of DZ type, the class of Hamiltonian PDEs of DZ type, and the class of Poisson operators are invariant under the Miura-type transformations (note however that the disk $B$, giving the ring~$\cS(u)=\mcO(B)$, may change under these transformations). In particular, after the Miura-type transformation~\eqref{Miuratypetrans}, a Poisson operator $\cP$ is transformed to 
$$
\widetilde{\cP}=\sum_{k,\ell\ge 0}(-1)^{\ell}\frac{\p \tu}{\p u_{k}}\circ \p_x^k \circ \cP 
\circ \p_x^{\ell}\circ \frac{\p \tu}{\p u_{\ell}}.
$$
We will say that the Miura-type transformation~\eqref{Miuratypetrans} is {\em close to identity} if $M^{[0]}(u)=u$.

\medskip

\subsection{Central invariant of a pair of compatible Poisson operators}

When Poisson operators $\cP_1,\cP_2$ are compatible, we call $\cP_2+\lambda \cP_1$ the {\it Poisson pencil} associated to~$\cP_1, \cP_2$. Consider the dispersionless limits of the Poisson operators $\cP_1,\cP_2$:
$$
\cP_i^{[0]}=g_i(u)\p_x+\frac12 g_i'(u)u_1,\quad i=1,2.
$$
Suppose that $g_1(u)\ne 0$ and $g_2(u)\ne\lambda g_1(u)$ for any $\lambda\in\mbC$. The associated {\em canonical coordinate} $r=r(u)$ is defined by $r(u):=g_2(u)/g_1(u)$. In the $r$-coordinate, the pencil of metrics associated to $\cP_2+\lambda \cP_1$ reads $r g(r)+\lambda g(r)$, with the metric $g(r):=(g_1(u)r'(u)^2)|_{u=u(r)}$. The equivalence class of Poisson pencils $\cP_2+\lambda \cP_1$ with the fixed dispersionless limit $\cP_2^{[0]}+\lambda \cP_1^{[0]}$ up to Miura-type transformations that are close to identity is characterized by the {\it central invariant} \cite{DLZ06, LZ05}
$$
c(r):=\frac{r'(u)^2}{3 g(r)^2} ( A_{2,3;2}-r A_{2,3;1})|_{u\mapsto u(r)}.
$$
Here $A_{2,3;i}$ is the coefficient in \eqref{PoissonOperator} for the operator $\cP_i$, $i=1,2$.

\medskip

Note that, by~\cite[Theorem~1.8]{DLZ06}, the pair of Poisson operators~$(\cP_1,\cP_2)$ is quasi-trivial, meaning that it can be reduced to the dispersionless part by a quasi-Miura transformation, and moreover this quasi-Miura transformation has the form
\begin{gather}\label{eq:quasi-Miura}
u\mapsto v=u+\sum_{k\ge 1}\eps^k F_k(u,\ldots,u_{[3k/2]}),\quad F_k(u,\ldots,u_{[3k/2]})\in\cA_u[u_x^{-1}],\quad\deg F_k=k.
\end{gather}
By \cite[Corollary~1.9]{DLZ06}, this quasi-Miura transformation automatically reduces any bihamiltonian PDE~\eqref{eq:bihamiltonian PDE} to the dispersionless part. This immediately implies the following important corollary.
\begin{cor}\label{corollary:from P1P2 to Q}
Consider a PDE $\frac{\d u}{\d t}=Q$ that is bihamiltonian with respect to a pair of compatible Poisson operators $(\cP_1,\cP_2)$ satisfying $\cP_1^{[0]}\ne 0$ and $\cP_2^{[0]}\ne\lambda \cP_1^{[0]}$ for any $\lambda\in\mbC$. Then~$Q$ is uniquely determined by $Q|_{\eps=0}$ and the operators $\cP_1,\cP_2$.
\end{cor}

\medskip

It was shown in~\cite{DLZ06} that under a different choice of a Poisson pencil
$$
\widetilde{\cP}_1=c \cP_2 +d \cP_1,\quad \widetilde{\cP}_2=a \cP_2 +b \cP_1,\quad ad-bc \neq 0,
$$
the canonical coordinate~$\tr$ for $\widetilde{\cP}_2+\lambda \widetilde{\cP}_1$ is given by $\tr=(ar+b)/(cr+d)$, and the pair of functions $(g(r),c(r))$ is transformed to
\beq\label{eq:change of g and c}
(\tg(\tr),\tc(\tr))=\biggl(\frac{(ad-bc)^2}{(cr+d)^3}g(r),\frac{cr+d}{ad-bc}c(r)\biggr)\bigg|_{r=r(\tr)}.
\eeq

\medskip

\begin{exa}{\ }
\begin{itemize}
\item For the compatible Poisson operators \eqref{KdVbihamiOp} of the KdV equation, we have 
$$
(g(r),c(r))= \lb1,\frac{1}{24}\rb.
$$

\smallskip

\item For the compatible Poisson operators \eqref{CHbihamiOp} of the CH equation, we have 
$$
(g(r),c(r))= \lb1,\frac{r}3\rb.
$$
\end{itemize}
\end{exa}

\medskip

\subsection{Perturbations of the RH hierarchy}

The following system of evolutionary PDEs of DZ type
\beq\label{perturbofRH}
u_{t_d}=\frac{u^d}{d!}u_x + O(\e),\quad d\ge 1,
\eeq
is called a {\it perturbation of the RH hierarchy}.
Such a perturbation is called {\em integrable} if the PDEs in~\eqref{perturbofRH} pairwise commute.
The RH hierarchy~\eqref{eq:RH hierarchy} is obviously integrable, we will also equip it with the flow $u_{t_0}=u_x$. 

\medskip

\begin{Def}[\cite{DLYZ16}]
A Hamiltonian integrable hierarchy of PDEs of DZ type
\beq\label{tausymhamiltonianhierarchy}
\frac{\p u}{\p t_n}=\cP\Bigl(\frac{\delta \oh_n}{\delta u}\Bigr),\quad n\ge 0,
\eeq
equipped with Hamiltonian densities $h_n$, is called a {\em $\tau$-symmetric Hamiltonian perturbation of the RH hierarchy} if $h_n|_{\e=0}=\frac{u^{n+2}}{(n+2)!}$, $\cP|_{\e=0}=\p_x$, the flow $\frac{\p}{\p t_0}$ is given by $\frac{\p u}{\p t_0}=u_x$, and the following two conditions are satisfied:
\begin{itemize}
\item[(i)] $\tau$-symmetry: 
$$
\frac{\p h_{p-1}(u;u_1,\dots;\e)}{\p t_q}=\frac{\p h_{q-1}(u;u_1,\dots;\e)}{\p t_p},\quad p,q\geq 0,
$$
where $h_{-1}:=u$;
\item[(ii)] $H_{-1}:=\int h_{-1} dx$ is a Casimir of the Poisson operator $\cP$.
\end{itemize}
\end{Def}

\medskip

Note that the Hamiltonian densities $h_n$ of a $\tau$-symmetric Hamiltonian perturbation of the RH hierarchy are uniquely determined by the Hamiltonians $\oh_n$ and the Poisson operator~$\cP$ (see, e.g.,~\cite[Remark~3.3]{BR25}). Moreover, all the Hamiltonians are uniquely determined by the Hamiltonian $\oh_1$ and the Poisson operator~$\cP$~\cite{LZ06}.

\medskip

The property of being a $\tau$-symmetric Hamiltonian perturbation of the RH hierarchy is preserved under {\em normal Miura-type transformations} that by definition have the form
$$
u\mapsto w=u+\e^2\p_x^2\biggl(\thin\sum_{k\geq 0}\e^k A_k\biggr),\quad A_k\in\mathcal{A}^{[k]}_u.
$$
Namely, the $\tau$-symmetric densities of the transformed hierarchy are given by 
$$
\widetilde{h}_{p}(w,w_1,\dots;\e)=h_{p}(u,u_1,\dots;\e)+\e^2\p_x\p_{t_{p+1}}\biggl(\thin\sum_{k\geq 0}\e^k A_k\biggr).
$$

\medskip

For a tuple $\od=(d_1,\ldots,d_n)\in\mbZ_{\ge 0}^n$, let $w_{\od}:=\prod_{i=1}^n w_{d_i}$. The following lemma was stated in~\cite{DLYZ16} with a partial proof, while a complete proof was given in~\cite{BR25}. 

\begin{lem}[\cite{BR25, DLYZ16}]\label{standardformthm}
For any $\tau$-symmetric Hamiltonian perturbation~\eqref{tausymhamiltonianhierarchy} of the RH hierarchy, there exists a unique normal Miura-type transformation that transforms~\eqref{tausymhamiltonianhierarchy} to the {\em generalized standard form}
\beq\label{standardformhierarchy-1}
\frac{\p w}{\p t_n}=\p_x\Bigl(\frac{\delta \widetilde{H}_n}{\delta w}\Bigr),\quad\widetilde{H}_n=\int\Big(\frac{w^{n+2}}{(n+2)!}+O(\e)\Big)dx,\quad n\geq 0,
\eeq
where 
\beq\label{standardformhierarchy-2}
\widetilde{H}_1=\int\Big(\thin\frac{w^3}{6}-\e^2\frac{a_0}{24}w_1^2+\sum_{k\geq 4}\e^{k}\sum_{\substack{\mu\in\mathcal{P}_{k},\,l(\mu)\ge 2\\ \mu_1=\mu_2,\,m_1(\mu)=0}}\alpha_{\mu} w_{\mu}\Big)dx,\quad \alpha_\mu\in\mbC.
\eeq
\end{lem}

\medskip

Denote by $\Mi^{(2)}$ the group of Miura-type transformations~\eqref{Miuratypetrans} that are close to identity and that satisfy $M^{[1]}=0$. This group contains the group of normal Miura-type tranformations. The following statement is slightly stronger than the previous lemma.

\begin{lem}[see Remark~3.3 and Proposition~3.4 in~\cite{BR25}]\label{standardformmiura}
For any $\tau$-symmetric Hamiltonian perturbation~\eqref{tausymhamiltonianhierarchy} of the RH hierarchy, there exists a unique Miura-type transformation from the group $\Mi^{(2)}$ that 
transforms~\eqref{tausymhamiltonianhierarchy} to the generalized standard form~\eqref{standardformhierarchy-1}--\eqref{standardformhierarchy-2}. This Miura-type transformation is necessarily normal. 
\end{lem}

\medskip

\begin{conj}[\cite{DLYZ16}]\label{conjecture:DLYZ}
Consider the generalized standard form of a $\tau$-symmetric Hamiltonian perturbation of the RH hierarchy.
\begin{enumerate}
\item If $a_0=0$, then $\alpha_\mu=0$ for all $\mu$.

\smallskip

\item If $a_0\ne 0$, then the coefficients $a_0,\alpha_{(2^2)}, \alpha_{(2^3)},\ldots$ uniquely determine all the other coefficients $\alpha_{\mu}$; and also the coefficients $\alpha_\mu$ are zero if $|\mu|$ is odd. Moreover, the coefficients $a_0\in\mbC^*$ and $\alpha_{(2^2)}, \alpha_{(2^3)},\ldots\in\mbC$ can be arbitrary.
\end{enumerate}
\end{conj}

\medskip

This conjecture was verified~\cite{DLYZ16} at the approximation up to $O(\e^{13})$. In particular, the authors of~\cite{DLYZ16} found the following formula in the case $a_0\ne 0$:
\begin{align}
\tH_1=&\int\biggl(\thin\frac{w^3}{6}-\e^2\frac{a_0}{24}w_1^2+\e^4\alpha_{(2^2)}w_2^2+\e^6\Bigl(\alpha_{(2^3)}w_2^3-\frac{240\alpha_{(2^2)}^2}{7a_0}w_3^2\Bigr)\notag\\
&+\e^8\Bigl(\alpha_{(2^4)}w_2^4-\frac{2376\alpha_{(2^2)}\alpha_{(2^3)}}{7a_0}w_2w_3^2+\Big(\frac{a_0\alpha_{(2^3)}}{35}+\frac{8640\alpha_{(2^2)}^3}{7a_0^2}\Big)w_4^2\Bigr)+O(\e^{10})\biggr)dx.\label{eq:tH1 up to e12}
\end{align}
The fact that the coefficients $\alpha_{(2^2)}, \alpha_{(2^3)},\ldots\in\mbC$ with $a_0\in\mbC^*$ can be arbitrary 
was proved in~\cite[Theorem~3.9]{BR25} (see the details in Section~\ref{subsection:CohFTs,FCohFTs,DRhierarchies}).

\medskip

Following \cite{ALM0, ALM} (see~also~\cite{BR25}), let us now consider evolutionary PDEs of DZ type of the form 
$$
\frac{\p u}{\p t} = \p_x \omega, \quad \omega=\frac{u^2}{2}+\sum_{k\ge 1}\e^k\sum_{\lambda\in\cP_k}c_\lambda(u)u_\lambda,\quad  c_\lambda(u)\in\cS(u).
$$ 
A Miura-type transformation that is close to identity preserves this class of PDEs if and only if it has the form
\beq\label{conserve-preserve-miura}
u\mapsto w=u+\e\p_x f,\quad f\in\hcA^{[0]}_u.
\eeq
Moreover, the functions $c_k(u)$, $k\ge 1$, are invariant under these transformations, and they are called the {\em Miura invariants}.

\medskip

\begin{exa}{\ }
\begin{itemize}
\item Writing the CH equation in the form
\beq\label{CHequationform2}
v_t=(1-\eps^2\d_x^2)^{-1}\d_x\Bigl(\thin\frac{v^2}{2}-\e^2 \Bigl(\thin\frac{v_x^2}{6}+\frac{v v_{xx}}{3}\Bigr)\Bigr),  
\eeq
we immediately compute its Miura invariants: $c_{2g}(v)=\frac{2}{3}v$.

\smallskip

\item Writing the DP equation in the form
\beq\label{DPequationform2}
v_t=(1-\eps^2\d_x^2)^{-1}\d_x\Bigl(\thin\frac{v^2}{2}-\e^2 \Bigl(\thin\frac{v_x^2}{4}+\frac{v v_{xx}}{4}\Bigr)\Bigr), 
\eeq
we immediately compute its Miura invariants: $c_{2g}(v)=\frac{3}{4}v$.
\end{itemize}
\end{exa}

\medskip

Consider now integrable hierarchies of scalar conservation laws that are perturbations of the RH hierarchy:
\beq\label{conservelawhierarchy}
\frac{\p u}{\p t_n} = \p_x \omega_n, \quad \omega_n=\frac{u^{n+1}}{(n+1)!}+O(\e)\in\hcA^{[0]}_u,\quad  n\geq 0.
\eeq 
By definition, the {\em Miura invariants} of the perturbation are the Miura invariants of the flow $\frac{\p}{\p t_1}$. Note that all the flows $\frac{\d}{\d t_n}$ are uniquely determined by the flow $\frac{\d}{\d t_1}$~\cite{LZ06}. It was proved in~\cite{ALM0, ALM} that for any hierarchy of the form~\eqref{conservelawhierarchy} there exists a unique Miura-type transformation of the form~\eqref{conserve-preserve-miura} that transforms~\eqref{conservelawhierarchy} to the {\em normal form}:
\beq\label{ALMnormalform}
\frac{\p w}{\p t_n}=\p_x\widetilde{\omega}_n,\quad 
\widetilde{\omega}_1=\frac{w^2}{2}+\e a(w)w_x
+\sum_{k\geq 2}\e^{k} \sum_{\lambda\in\mathcal{P}_{k}\atop m_1(\lambda)=0} c_{\lambda}(w)w_{\lambda},\quad a(w),c_\lambda(w)\in\mathcal{S}(w).
\eeq

\medskip

\begin{conj}[\cite{ALM}]\label{conjecture:ALM}
Consider the normal form of an integrable perturbation of the RH hierarchy~\eqref{conservelawhierarchy} and suppose that $a(w)=0$ and $c_2(w)\ne 0$. Then the following statements are true:
\begin{enumerate}
\item $c_{\lambda}=0$ for all $\lambda$ with odd $|\lambda|$;

\smallskip

\item the coefficients $c_{2}(w),c_{4}(w),\dots$ uniquely determine all the other coefficients~$c_{\lambda}(w)$;

\smallskip

\item the coefficients $c_{2}(w),c_{4}(w),\dots$ are all free functional parameters.
\end{enumerate}
\end{conj}

\medskip

We call Conjecture~\ref{conjecture:ALM} the {\em ALM conjecture}. Parts~1 and~2 of the ALM conjecture were verified in~\cite{ALM} at the approximation up to $O(\e^7)$. The fact that the coefficients $c_{2},c_{4},\dots$ can be arbitrary complex constants with $c_2\ne 0$ 
was proved in~\cite[Theorem~4.4]{BR25} (see the details in Section~\ref{subsection:CohFTs,FCohFTs,DRhierarchies}).

\medskip

Any $\tau$-symmetric Hamiltonian perturbation of the RH hierarchy written in the generalized standard form has the normal form~\eqref{ALMnormalform}. In particular, from~\eqref{eq:tH1 up to e12} we obtain
\begin{align*}
\tomega_1=\frac{\delta\tH_1}{\delta w}=&\frac{w^2}{2}+\frac{a_0}{12} \e^2 w_{2}+2 \alpha_{(2^2)} \e^4 w_{4}+\e^6 \biggl(\frac{480 \alpha_{(2^2)}^2}{7 a_0}w_{6}+6\alpha_{(2^3)} w_{3}^2+6 \alpha_{(2^3)} w_{4} w_{2}\biggr)\notag\\
&+\e^8  \biggl(\biggl(\frac{17280 \alpha_{(2^2)}^3}{7 a_0^2}+\frac{2}{35} a_0 \alpha_{(2^3)} \biggr)w_{8}+\frac{9504 \alpha_{(2^3)} \alpha_{(2^2)}}{7 a_0}w_{4}^2+\frac{4752 \alpha_{(2^3)} \alpha_{(2^2)}}{7 a_0}w_{6}w_{2}\notag\\
&\hspace{1.2cm}+\frac{14256 \alpha_{(2^3)} \alpha_{(2^2)}}{7 a_0}w_{3} w_{5}+12 \alpha_{(2^4)} w_{4}w_{2}^2+24 \alpha_{(2^4)} w_{3}^2 w_{2}\biggr)\notag\\
&+O(\e^{10}).
\end{align*}
In~\cite[Theorem 4.8]{BR25}, Buryak--Rossi proved that the ALM conjecture implies Part~2 of 
Conjecture~\ref{conjecture:DLYZ}.

\medskip

\section{Integrable systems associated to rank-$1$ F-CohFTs}\label{review}

In this section, we study relations between integrable hierarchies and the topology of the Deligne--Mumford moduli spaces of stable algebraic curves $\Mgnbar$ via F-CohFTs.

\medskip

\subsection{CohFTs, F-CohFTs, and the DR hierarchies}\label{subsection:CohFTs,FCohFTs,DRhierarchies}

Cohomological field theories (CohFTs) and F-CohFTs were introduced in \cite{KM} and \cite{BR21}, respectively. For the general definitions, and a discussion on their differences, the reader is referred for instance to~\cite[Section 3.1]{ABLR23}. In this paper, we will only consider CohFTs and F-CohFTs of rank~$1$, with phase space $V=\mbC$ and unit $e=1$. In the case of CohFTs, the metric~$\eta$ will be always given by $\eta(e,e)=1$.

\medskip

So, for this paper, a \emph{CohFT} is a family of cohomology classes 
$$
c_{g,n}\in H^{\mathrm{e}}(\oM_{g,n}),\quad 2g-2+n>0,
$$ 
that are invariant under permutations of the $n$ marked points and satisfy the following axioms:
\begin{enumerate}
\item [1)] $c_{0,3}=1$, $\pi^*c_{g,n} = c_{g,n+1}$ for the map $\pi\colon\oM_{g,n+1} \to \oM_{g,n}$ forgetting the last marked point;
\item [2)] $\sigma^* c_{g_1+g_2,n_1+n_2}= c_{g_1,n_1+1} \times c_{g_2,n_2+1}$ for the map $\sigma\colon \oM_{g_1,n_1+1}\times \oM_{g_2,n_2+1} \to \oM_{g_1+g_2,n_1+n_2}$ joining two stable curves at their last marked points to form a nodal curve;
\item [3)] $\tsigma^* c_{g+1,n}= c_{g,n+2}$ for the map $\tsigma\colon \oM_{g,n+2}\to \oM_{g+1,n}$ gluing last two marked points on a stable curve to form a nodal curve.
\end{enumerate}
The definition of an \emph{F-CohFT} differs from the one of a CohFT in the following:
\begin{itemize}
\item[a)] We require $n\ge 1$, while invariance is under permutations of the last $n-1$ marked points only.
\item[b)] We do not require Axiom~3.
\end{itemize}
So restricting a CohFT to moduli spaces with $n\ge 1$ gives an F-CohFT.

\medskip

The \emph{trivial CohFT} is given by $c_{g,n}^{\triv}:=1\in H^0(\oM_{g,n})$. The corresponding F-CohFT will be denoted by the same symbol. For any formal power series $R(z)\in 1+z\mbC[[z]]$ satisfying $R(z)R(-z)=1$, there is a CohFT denoted by $\{R.c^{\triv}_{g,n}\}$ and defined by the standard Givental formula for CohFTs, as described for instance in \cite{PPZ15}. Moreover, an arbitrary CohFT has this form~\cite{Tel12}. The CohFT $\{R.c^{\triv}_{g,n}\}$ can be alternatively described as follows:
\beq\label{sigmaCohFT}
R.c^{\triv}_{g,n}=e^{\sum_{i\geq 1} \sigma_{2i-1} \gamma_{2i-1}}, \quad \sigma_{2i-1}\in\CC,
\eeq 
where $R(z) = e^{\sum_{i\ge 1} r_{2i-1}z^{2i-1}}$ with $r_{2i-1}=\frac{B_{2i}}{(2i)!}\sigma_{2i-1}$ and $\gamma_j$ is the $j$th component 
of the Chern character of the Hodge bundle~$\mbE$. In~\cite[Section~4]{ABLR23}, Arsie--Buryak--Lorenzoni--Rossi defined an F-CohFT $\{R.c^{\triv}_{g,n+1}\}$ for an arbitrary $R(z)\in 1+z\mbC[[z]]$ by a formula similar to Givental's formula where the sum runs again over stable trees only.

\medskip

For any F-CohFT $\{c_{g,n+1}\}$, we have the associated \emph{DR hierarchy}
$$
\frac{\d u}{\d t_d}=\d_x P_d,\quad d\ge 0,
$$
where 
$$
P_d:=\hspace{-2.5mm}\sum_{\substack{g,n\ge 0\\\od=(d_1,\ldots,d_n)\in\mbZ_{\ge 0}^n\\d_1+\ldots+d_n=2g}}\hspace{-2.5mm}\frac{\e^{2g}}{n!}\Coef_{a_1^{d_1}\cdots a_n^{d_n}}\lb\int_{\oM_{g,n+2}}\hspace{-0.5cm}\psi_2^d\lambda_g\DR_g\lb-\sum a_i,0,a_1,\ldots,a_n\rb c_{g,n+2}\rb u_{\od}.
$$
Here $\psi_i\in H^2(\oM_{g,n},\mbQ)$, for $1\leq i\leq n$, denotes the first Chern class of the $i$-th \emph{tautological line bundle} on $\oM_{g,n}$ whose fiber over a curve is the cotangent line to the curve at the $i$-th marked point, $\lambda_g \in H^{2g}(\oM_{g,n},\mbQ)$ is the top Chern class of the Hodge bundle~$\mbE$, and $\DR_g(a_1,\ldots,a_n) \in H^{2g}(\oM_{g,n},\mbQ)$ is the double ramification cycle, a cohomology class, polynomial of degree $2g$ in the variables $a_1,\ldots,a_n$, which represents a compactification by relative stable maps to $\CP^1$ of the locus of smooth curves whose marked points form the support of a principal divisor with multiplicities $a_1,\ldots,a_n$. More details on these classes and the DR hierarchy construction can be found in the original paper \cite{Buryak15} introducing the DR hierarchy and its F-CohFT counterpart \cite{ABLR21}.
 
\medskip

The DR hierarchy is integrable. It is straightforward from the above definition in terms of intersection numbers that the differential polynomials $P_d$ have the form $P_d=\frac{u^{d+1}}{(d+1)!}+O(\e^2)$. So the DR hierarchy is an integrable perturbation of the RH hierarchy of the form~\eqref{conservelawhierarchy}. Moreover, $P_1$ has the form~\cite[Theorem~4.4]{BR25}
$$
P_1=\frac{u^2}{2}+\sum_{g\geq 1}\e^{2g} \sum_{\lambda\in\mathcal{P}_{2g}\atop m_1(\lambda)=0} c_{\lambda}u_{\lambda},\quad c_\lambda\in\mbC,
$$
which implies that the DR hierarchy is in the normal form. By~\cite[Theorem~3]{ABLR21}, the DR hierarchy is endowed with conserved quantities $\og_d=\int g_d dx$, $d\ge 0$, given~by
$$
g_d:=\hspace{-1mm}\sum_{\substack{g,n\ge 0\\\od=(d_1,\ldots,d_n)\in\mbZ_{\ge 0}^n\\d_1+\ldots+d_n=2g}}\hspace{-1mm}\frac{\e^{2g}}{n!}\Coef_{a_1^{d_1}\cdots a_n^{d_n}}\Biggl(\int_{\oM_{g,n+1}}\hspace{-0.3cm}\psi_1^d\lambda_g\DR_g\lb-\sum a_i,a_1,\ldots,a_n\rb c_{g,n+1}\Biggr) u_{\od}.
$$
The differential polynomials $g_d$ have the form $g_d=\frac{u^{d+2}}{(d+2)!}+O(\e^2)$. 

\medskip

Consider now the F-CohFT $\{R.c^\triv_{g,n+1}\}$, where $R(z)=\exp(\sum_{i\ge 1}r_i z^i)$. In \cite[Theorem~4.4]{BR25}, Buryak--Rossi proved that the coefficients $c_{\lambda}$ are polynomials in $r_1,\dots, r_{g+\ell(\lambda)-2}$ with rational coefficients. For example, via a direct but tedious computation we obtain
\begin{align*}
P_1=&\frac{u^2}{2}+\eps^2\frac{u_{xx}}{12}+\eps^4\left(\frac{r_1}{60} u_{xxxx}+\frac{r_2}{48}u_{xx}^2\right)\\
&+\epsilon ^6 \biggl(\biggl(\frac{r_1^2}{210}+\frac{r_2}{1120}\biggr) u_6+\biggl(\frac{r_1^3}{60} + \frac{7 r_1 r_2}{240}+\frac{r_3}{288}\biggr) u_4 u_2 +\\
&\hspace{10mm}+\biggl(\frac{r_1^3}{60}+ \frac{13 r_1 r_2}{720}+\frac{r_3}{288}\biggr) u_3^2+\biggl(\frac{r_4}{756}+\frac{611}{60480}r_2^2+\frac{r_2r_1^2}{7560}\biggr)u_2^3\biggr)+O(\eps^8).
\end{align*}
In \cite[Theorem~4.4]{BR25}, Buryak--Rossi proved that for $g\ge 2$ the coefficients $c_{2g}, c_{(2^g)}, c_{(4,2^{g-2})}$ have the form
\begin{align*}
&c_{2g} = \frac{2g}{4^g}\frac1{(2g+1)!!} r_{g-1}
+\widetilde{T}_{2g}(r_1,\dots,r_{g-2}),\\
&c_{(2^g)} = (3g-1)(2g-1)\frac{|B_{2g}|}{(2g)!}r_{2g-2}
+\widetilde{T}_{(2^g)}(r_1,\dots,r_{2g-3}),\\
&c_{(4, 2^{g-2})} = \frac{(3g-2)|B_{2g}|}{8(2g-3)!} r_{2g-3}
+\widetilde{T}_{(4,2^{g-2})}(r_1,\dots,r_{2g-4}),
\end{align*}
where $\widetilde{T}_{2g},\widetilde{T}_{(2^g)},\widetilde{T}_{(4, 2^{g-2})}$ are polynomials with rational coefficients. This implies that the family of DR hierarchies corresponding to the F-CohFTs $\{R.c_{g,n+1}\}$ can be parameterized either by the constants $c_4,c_6,\ldots$ or by the constants $c_{(2^g)}$, $c_{(4,2^{g-2})}$, $g\ge 2$. In particular, this shows that the coefficients $c_2,c_4,\ldots$ in~\eqref{ALMnormalform} (assuming $a(w)=0$) can be arbitrary constants with $c_2\ne 0$ (one has to use rescaling of $\eps$, because the coefficient~$c_2$ is fixed in the DR hierarchy). As an example of the parameterization using the constants $a_g:=c_{(4,2^{g-2})}$ and $b_g:=c_{(2^g)}$, we have
\begin{align}
P_1=&\frac{u^2}{2}+\frac{\eps^2}{12} u_{xx}+\eps^4(a_2 u_{xxxx}+b_2 u_{xx}^2)\nn\\
&+\epsilon ^6 \Bigl(\Bigl(\thin\frac{120}{7}a_2^2+\frac{3}{70}b_2\Bigr) u_6+(a_3-32 a_2 b_2) u_3^2+b_3 u_2^3+a_3 u_4 u_2\Bigr)+O(\eps^{8}).\label{eq:P1 from a and b}
\end{align}

\medskip

If our F-CohFT is associated to a CohFT, then the DR hierarchy is Hamiltonian with Poisson operator $\d_x$ and Hamiltonians $\og_d$, i.e., $P_d=\frac{\delta\og_d}{\delta u}$. Moreover, it is a $\tau$-symmetric Hamiltonian perturbation of the RH hierarchy with $h_{d-1}=\frac{\delta \og_{d}}{\delta u}$ and $\oh_d=\og_d$, for $d\geq 0$, and it is also in the generalized standard form, as shown in~\cite{BDGR18},
$$
\og_1=\int\Big(\frac{u^3}{6}-\e^2\frac{u_x^2}{24}+\sum_{g\geq 2}\e^{2g}\sum_{\substack{\mu\in\mathcal{P}_{2g},\,l(\mu)\ge 2\\ \mu_1=\mu_2,\,m_1(\mu)=0}}\alpha_{\mu} u_{\mu}\Big)dx,\quad \alpha_\mu\in\mbC.
$$
For the CohFT~\eqref{sigmaCohFT}, in~\cite[Theorem~3.9]{BR25} Buryak--Rossi proved that the coefficients~$\alpha_{\lambda}$ are polynomials in $\sigma_1,\dots,\sigma_{2\lfloor\frac{g+\ell(\lambda)}{2}\rfloor-3}$ with rational coefficients. For example, we have
\begin{equation*}
\og_{1}=\int\biggl(\frac{u^3}{6}-\frac{\eps^2}{24}u_x^2+\frac{\eps^4}{120} r_1 u_{xx}^2+
\eps^6\Bigl(\Bigl(\thin\frac{r_1^3}{360} +\frac{r_3}{1728}\Bigr)u_{xx}^3-\frac{r_1^2}{420}u_{xxx}^2\Bigr)+O(\eps^8)\biggr)dx.
\end{equation*}
Moreover, it was proved in~\cite[Theorem~3.9]{BR25} (cf.~\cite{BDGR18, FP99}) that the coefficients $\alpha_{(2^g)}$, $g\ge 2$, 
have the more precise form 
\beq\label{trigorelation}
\alpha_{(2^g)}=(-1)^g\frac{(3g-2)|B_{2g-2}||B_{2g}|}{4g ((2g-2)!)^2}\sigma_{2g-3}+T_g(\sigma_1,\dots,\sigma_{2g-5}),\quad T_g\in \QQ[\sigma_1,\dots,\sigma_{2g-5}].
\eeq
This shows that the coefficients $a_0,\alpha_{(2^2)},\alpha_{(2^3)},\ldots$ in~\eqref{standardformhierarchy-2} can be arbitrary constants with $a_0\ne 0$ (one has to use rescaling of $\eps$, because the coefficient~$a_0$ is fixed in the DR hierarchy).

\medskip

One can associate to an F-CohFT the \emph{Dubrovin--Zhang (DZ) hierarchy}, which is also a perturbation of the RH hierarchy. 
If the F-CohFT comes from a CohFT, then the DZ hierarchy is a $\tau$-symmetric Hamiltonian perturbation of the RH hierarchy. It is called the \emph{Hodge hierarchy} of a point. The DR hierarchy and the DZ hierarchy are equivalent using a Miura-type transformation that is close to identity~\cite{BS24,BSS25}. The following conjecture was proposed in~\cite{DLYZ16} together with Conjecture~\ref{conjecture:DLYZ}.

\medskip

\noindent {\bf Hodge Universality Conjecture}~(\cite{DLYZ16}).
{\it Any $\tau$-symmetric Hamiltonian perturbation of the RH hierarchy with $a_0=1$ in the generalized standard form is equivalent to the Hodge hierarchy of a point by a normal Miura transformation.
}

\medskip  

\noindent Since the coefficients $a_0,\alpha_{(2^2)},\alpha_{(2^3)},\ldots$ in~\eqref{standardformhierarchy-2} can be arbitrary constants with $a_0\ne 0$, the Hodge Universality Conjecture follows from Conjecture~\ref{conjecture:DLYZ}.

\medskip

\subsection{Local DR hierarchy}
Consider an integrable perturbation of the RH hierarchy. Any infinitesimal symmetry of this perturbation of the form $\frac{\p u}{\p t}=f(u)u_x+O(\e)$, 
$f(u)\in\cS(u)$ and $f(u)\ne {\rm const}$, is uniquely determined by the leading term~\cite{LZ06}. We denote this flow (if exists) by $\frac{\p}{\p t_f}$. The following lemma generalizes the construction of {\it local Hodge hierarchy}~\cite[formula~(317)]{YZ2023}.

\medskip

\begin{lem} 
Consider the DR hierarchy corresponding to an F-CohFT $\{c_{g,n+1}\}$. The following statements hold:
\begin{enumerate}
\item [1)] The differential polynomials $P_d, g_d$ (giving the flows and the conserved quantities of the DR hierarchy) have the form
\begin{align*}
&P_d=\frac{u^{d+1}}{(d+1)!}+\sum_{g\ge 1}\e^{2g}\sum_{\lambda\in\cP_{2g}}\sum_{1\le k\le\min(g-1+l(\lambda),d)}p_{\lambda,k}\frac{u^{d-k}}{(d-k)!}u_\lambda,&& p_{\lambda,k}\in\mbC, \, d\ge0,\\
&g_d=\frac{u^{d+2}}{(d+2)!}+\sum_{g\ge 1}\e^{2g}\sum_{\lambda\in\cP_{2g}}\sum_{0\le k\le\min(g-2+l(\lambda),d)}q_{\lambda,k}\frac{u^{d-k}}{(d-k)!}u_\lambda,&& q_{\lambda,k}\in\mbC, \, d\ge 0.
\end{align*}

\smallskip

\item[2)] 
\begin{enumerate}
\item [a)] For any $f(u)\in\cS(u)$, there exists a unique infinitesimal symmetry of the DR hierarchy of the form $\frac{\p u}{\p t_f}=f(u)u_x+O(\e)$. Moreover, this infinitesimal symmetry is given by 
\begin{align}
& \frac{\d u}{\d t_f}=\d_x P_f,\quad\text{where} \label{localDRh}\\
& P_f:=F(u)+\sum_{g\ge 1}\e^{2g}\sum_{\lambda\in\cP_{2g}}\sum_{1\le k\le g-1+l(\lambda)}p_{\lambda,k}f^{(k)}(u)u_\lambda,\quad F'(u)=f(u).\notag
\end{align}

\smallskip

\item[b)] For any $h(u)\in\cS(u)$, there exists a unique conserved quantity of the DR hierarchy of the form $\int\lb h(u)+O(\eps)\rb dx$. Moreover, this conserved quantity is given by $\int g_h dx$, where
$$
g_h:=h(u)+\sum_{g\ge 1}\e^{2g}\sum_{\lambda\in\cP_{2g}}\sum_{0\le k\le g-2+l(\lambda)}q_{\lambda,k}h^{(k+2)}(u)u_\lambda.
$$

\smallskip

\item[c)] The local functional $\int g_h dx$ is a conserved quantity of~\eqref{localDRh}. 
\end{enumerate}
\end{enumerate}
\end{lem}
\begin{proof}
{\em 1}. Since $\frac{\d P_b}{\d u}=P_{b-1}$, $b\ge 1$, and $P_0=u$, we have 
$$
P_d=\frac{u^{d+1}}{(d+1)!}+\sum_{g\ge 1}\e^{2g}\sum_{\lambda\in\cP_{2g}}\sum_{k=1}^d p_{\lambda,k}\frac{u^{d-k}}{(d-k)!}u_\lambda,\quad p_{\lambda,k}\in\mbC.
$$
By degree counting, the integral $\int_{\oM_{g,n+2}}\psi_2^k\lambda_g\DR_g\lb-\sum a_i,0,a_1,\ldots,a_n\rb c_{g,n+2}$ is zero unless $k\le g-1+n$. This implies that $p_{\lambda,k}$ vanishes unless $k\le g-1+l(\lambda)$. The claim for $g_d$ is proved by analogous arguments, using that $\frac{\d g_b}{\d u}=g_{b-1}$, $b\ge 0$, where $g_{-1}:=u$.

\medskip

{\em 2a)}. Let $c\in\mbC$ be the center of the open disk~$B\subset\mbC$ giving the ring $\cS(u)=\mcO(B)$, and consider the Taylor expansion of the function $f$, $f(u)=\sum_{d\ge 0}f_d\frac{(u-c)^d}{d!}$. Define $P_{d,c}:=P_d|_{u\mapsto u-c}$. Since $P_{d,c}=\sum_{k=0}^{d+1}\frac{(-c)^k}{k!}P_{d-k}$, where $P_{-1}:=1$, the flow $\frac{\p u}{\p t}=\p_x P_{d,c}$ is an infinitesimal symmetry of the DR hierarchy. Obviously, we have $P_f=\sum_{d\ge 0}f_d P_{d,c}$, and the flow $\frac{\p u}{\p t_f}=\p_x P_f$ is an infinitesimal symmetry of the DR hierarchy.

\medskip

{\em 2b)}. Similarly to 2a), write $h(u)=\sum_{d\ge 0} h_d\frac{(u-c)^d}{d!}$ and define $g_{d,c}:=g_d|_{u\mapsto u-c}$. Since $g_{d,c}=\sum_{k=0}^{d+2}\frac{(-c)^k}{k!}g_{d-k}$, where $g_{-2}:=1$, the local
functional $\int g_{d,c}dx$ is a conserved quantity of the DR hierarchy. The statement is proved by noticing that $g_h=\sum_{d\ge 0} h_d g_{d-2,c}$.

\medskip

{\em 2c)}. Notice that the flow $\frac{\p}{\p t_f}$ is an infinite linear combination of the flows $D_{\p_x(P_{d,c})}$.
\end{proof}

\medskip

The collection of the flows~\eqref{localDRh}, for all $f(u)\in\cS(u)$, is called the {\it local DR hierarchy}. The dispersionless limit of the local DR hierarchy is called the {\it local RH hierarchy} (cf.~also~\cite{YZ2023}).

\medskip

\section{Bihamiltonian $\tau$-symmetric hierarchies}\label{bihamiHodge}

In this section we prove a conjecture regarding bihamiltonian $\tau$-symmetry.

\medskip

\begin{Def}
A $\tau$-symmetric Hamiltonian perturbation of the RH hierarchy is called \emph{bihamiltonian} if the first flow has the form
$$
\frac{\d w}{\d t_1}=\cP_2\biggl(\frac{\delta\oh'}{\delta w}\biggr),\quad\oh'\in\hcF^{[0]}_w.
$$ 
where $(\cP_1=\cP,\cP_2)$ is a pair of compatible Poisson operators such that $\cP_2^{[0]}\ne\lambda \cP_1^{[0]}$ for any $\lambda\in\mbC$.
\end{Def}

\medskip

Consider a $\tau$-symmetric Hamiltonian perturbation of the RH hierarchy. We want to study when it is bihamiltonian. According to Lemma~\ref{standardformthm}, we can assume that it is in the generalized standard 
form~\eqref{standardformhierarchy-1}--\eqref{standardformhierarchy-2}. Suppose this perturbation is bihamiltonian. 

If $a_0 \neq 0$, then according to~\cite[Theorem~4.1]{Dubrovin06} we know that 
$$
\cP_1=\p_x,\quad \cP_2=\phi(w)\p_x+\frac12 \phi'(w)w_x+O(\e^2),\quad\phi(w)\ne\const,
$$
with 
\begin{gather}\label{eq:phi depending on alpha}
\phi(w)= 
\begin{cases}
	C_1 w+C_2,& \text{if $\alpha_{(2^2)}=0$},\\
	C_1 e^{-960 \alpha_{(2^2)}w/a_0^2}+C_2,& \text{if $\alpha_{(2^2)}\neq 0$},
\end{cases}
\end{gather}
where $C_1,C_2$ are constants and $C_1\neq 0$. Moreover, by \cite[Theorem~A.1]{Dubrovin06}, the Poisson operator $\cP_2$ up to $O(\eps^5)$ is uniquely determined by $a_0$ and $\phi(w)$. In particular, in the case $a_0\ne 0$, by a proper choice of the Poisson pencil, the metric $g(r)$ in the canonical coordinate and the central invariant $c(r)$ read~\cite{YZ2023}
\begin{gather*}
(g(r),c(r))=\Bigl(\,\frac{960 \alpha_{(2^2)}}{a_0}r +a_0,\frac1{24}\Bigr).
\end{gather*}

\medskip

The Hodge hierarchy associated to the CohFT $c_{g,n}^{\rm special}$ (see~\eqref{eq:CohFT for Volterra}, \eqref{specializationofHodge}) is called in~\cite{YZ2023} the {\it special-Hodge hierarchy}, which according to~\cite{DLYZ20, DLYZ16} is bihamiltonian. 
For the special-Hodge hierarchy, in the generalized standard form, we have $a_0=1$ and $\alpha_{(2^2)}=q/480$. A variant of the following theorem (see Theorem~\ref{bihamitausymm}$'$ below) was conjectured in~\cite{DLYZ16} (cf.~also~\cite{YZ2023}), and will be proved here using an argument from~\cite{YZ2023}. 

\medskip

\begin{thm}\label{bihamitausymm}
A $\tau$-symmetric Hamiltonian perturbation of the RH hierarchy, with $a_0=1$ in the generalized standard form, is bihamiltonian if and only if it is Miura equivalent to the special-Hodge hierarchy.
\end{thm}
\begin{proof}
It suffices to prove necessity. Consider our bihamiltonian $\tau$-symmetric perturbation of the RH hierarchy and the special-Hodge hierarchy, both in the generalized standard form. After an appropriate choice of the parameter $q$, the Hamiltonians of both hierarchies become the same at the approximation up to $O(\eps^5)$. Then one can choose an appropriate second Poisson operator for our $\tau$-symmetric perturbation 
(see formula~\eqref{eq:phi depending on alpha}) in such a way that the second Poisson operators of our hierarchies become the same at the approximation up to $O(\eps^5)$. In particular, the central invariants of the two pairs of compatible Poisson operators become the same. Therefore, there exists \cite{LZ05} a Miura-type transformation that is close to identity that relates the Poisson pairs of the two hierarchies. 
By Corollary~\ref{corollary:from P1P2 to Q}, this Miura-type transformation relates the flows $\frac{\d}{\d t_1}$ of the two hierarchies, and therefore the whole hierarchies as well.
\end{proof}

\medskip

We note that for the case when $a_0\neq 0$ we can rescale~$\e$ so that $a_0=1$. 

\medskip

If $a_0=0$, then from~\cite{Dubrovin06} we know that up to $O(\e^5)$ all the perturbation terms vanish. 
Let us show that for this case the generalized standard form is trivial. Indeed, the central invariant now is~0, so 
the corresponding Poisson pencil can~\cite{LZ05} be reduced to its dispersionless limit by a Miura-type transformation, which 
also reduces the generalized standard form to its dispersionless limit. Triviality then follows from Lemma~\ref{standardformthm}.

\medskip

By DR/DZ equivalence, Lemma~\ref{standardformmiura} and formula~\eqref{trigorelation}, the following theorem is equivalent to Theorem~\ref{bihamitausymm}.

\begin{manualtheorem}{\ref{bihamitausymm}$'$}
The Hodge hierarchy associated to a CohFT~$c_{g,n}$ is bihamiltonian if and only if $c_{g,n}=c_{g,n}^{\rm special}$ for some $q$.
\end{manualtheorem}

\medskip

\section{Bihamiltonian tests for the DR hierarchies for rank-$1$ F-CohFTs}\label{bihamitest}

In this section, we do a bihamiltonian test for the DR hierarchy corresponding to the F-CohFT 
$\{R(z).c_{g,n+1}^\triv\}$, $R(z)\in 1+z\mbC[[z]]$, with the first flow
\beq\label{BReq0}
\frac{\d u}{\d t_1}=\d_x P_1,\quad P_1=\frac{u^2}{2}+\eps^2\frac{u_{xx}}{12}+\sum_{g\geq 2}\e^{2g} \sum_{\lambda\in\mathcal{P}_{2g}\atop m_1(\lambda)=0} c_{\lambda}u_{\lambda},\quad c_\lambda\in\mbC.
\eeq
As we already mentioned in Section~\ref{subsection:CohFTs,FCohFTs,DRhierarchies}, this family of hierarchies can be parameterized using the constants $a_g:=c_{(4,2^{g-2})}$ and $b_g:=c_{(2^g)}$, $g\ge 2$, see for example~\eqref{eq:P1 from a and b}. 

\medskip

So suppose that the PDE~\eqref{BReq0} is bihamiltonian with respect to a pair of compatible Poisson operators $(\cP_1,\cP_2)$ such that 
\beq\label{eq:conditions for P1P2}
\cP_1^{[0]}\ne 0,\quad \cP_2^{[0]}\ne\lambda \cP_1^{[0]}\text{ for any $\lambda\in\mbC$}.
\eeq
We know that there exists a quasi-Miura transformation of the form~\eqref{eq:quasi-Miura} 
that reduces the pair $(\cP_1, \cP_2)$ and the flow $\frac{\d}{\d t_1}$ to the dispersionless part. For a Poisson operator of hydrodynamic type
$$
\cP(v) = \phi(v) \partial_x + \frac12 \phi'(v) v_x, \quad \phi(v) \ne 0,
$$
substituting the quasi-Miura transformation and requiring polynomiality for the coefficient of $\epsilon^4$, we get 
$$
768 b_2 \phi - 480 a_2 \phi' - \phi''=0.
$$
It follows that
$$
\phi(v)=p_1 \phi_1(v)+p_2 \phi_2(v),
$$
where $p_1, p_2$ are arbitrary constants and 
\beq\label{metric}
\phi_1(v)=e^{-(240 a_2+16\sqrt{\Delta})v},\quad
\phi_2(v)=\frac{1}{32 \sqrt{\Delta}}\Bigl(e^{-(240 a_2-16\sqrt{\Delta}) v}-e^{-(240 a_2+16\sqrt{\Delta}) v}\Bigr).
\eeq
Here 
$$
\Delta:=225a_2^2+3 b_2,
$$
and it is understood that when $\Delta=0$, $\phi_1(v)=e^{-240a_2 v}$ and $\phi_2(v)=v e^{-240a_2 v}$.

\medskip

Requiring for all $p_1,p_2$ the polynomiality for the coefficient of $\epsilon^6$, we get 
$$
a_3 = 1600 a_2^3 + 44 a_2 b_2, \quad b_3 = \frac{10240 a_2^2 b_2}{7} - \frac{128 b_2^2}{15}.
$$
Requiring  for all $p_1,p_2$ the polynomiality for the coefficient of $\epsilon^8$, we get 
$$
a_4 = \frac{2073600 a_2^3 b_2}{7}+\frac{5760 a_2 b_2^2}{7}, \quad b_4 = \frac{15488 b_2^3}{21}-\frac{101376 a_2^2 b_2^2}{7}.
$$
Requiring  for all $p_1,p_2$ the polynomiality for the coefficient of $\epsilon^{10}$, we get 
\begin{align}
&a_5=-\frac{40614912000 a_2^5 b_2}{77}-\frac{984821760}{77} a_2^3 b_2^2-\frac{5512832 a_2 b_2^3}{385}-\frac{5750784000000  a_2^7}{77},\nn\\
&b_5=-\frac{495452160000 a_2^6 b_2}{11}-\frac{5197824000}{11} a_2^4 b_2^2-\frac{376034652160}{11} a_2^2 b_2^3-\frac{26658304  b_2^4}{275}.\nn
\end{align}

\medskip

\begin{conj}\label{bihamDRconj}
For arbitrary $a_2, b_2\in\CC$, there exists a unique choice of $a_3, b_3, a_4, b_4$, $\dots$ such that equation~\eqref{BReq0} is bihamiltonian 
with $\cP_1, \cP_2$ satisfying the condition~\eqref{eq:conditions for P1P2}. Moreover, with this choice, $a_i$, $b_i$, $i\ge 3$, have a polynomial dependence on $a_2,b_2$.
\end{conj}

\medskip

The above bihamiltonian test together with the results in~\cite{ALM, BR25} (cf.~Section~\ref{review}) implies the validity of Conjecture~\ref{bihamDRconj} up to and including $O(\e^{10})$.

\begin{rmk}
We can also perform bihamiltonian test for~\eqref{ALMnormalform} with 
$n=1$, $a(w)\equiv 0$, $c_{\lambda}(w)$ being arbitrary constants and $c_{2}=\frac1{12}$.
A similar statement about a 2-parameter family can be given, which we are actually able to prove. 
See ``{\bf Note added}" at the end of Section~\ref{Miura-invariants-transf} for more details. 

\end{rmk}

\medskip

\begin{rmk}\label{remark:degree of ag and bg}
By~\cite[Theorem~4.4]{BR25}, the coefficient $c_\lambda$ in~\eqref{BReq0} is a polynomial in $r_i$ of degree $g+l(\lambda)-2$, where $\deg r_i:=i$. In particular, the degree of $a_g$ is $2g-3$ and the degree of $b_g$ is $2g-2$. Therefore, the rescaling $r_i\mapsto \theta^i r_i$ corresponds to the rescaling $c_\lambda\mapsto\theta^{g+l(\lambda)-2}c_\lambda$. On the other hand, this rescaling corresponds to the following rescalings
$$
u\mapsto \theta u,\quad t_1\mapsto\theta^{-1}t_1,\quad\eps^2\mapsto\theta\eps^2,
$$
which obviously does not change for~\eqref{BReq0} the property of being bihamiltonian. Assuming the above conjecture is true, this implies that
\begin{itemize}
\item the bihamiltonian DR hierarchies corresponding to the pairs $(a_2,b_2)$ and $(\theta a_2,\theta^2 b_2)$, $\theta\in\mbC^*$, are essentially the same: they are related by rescalings of $u$, $\eps$, and the times;

\smallskip

\item for the bihamiltonian DR hierarchies, $a_g$ is a polynomial in $a_2,b_2$ of degree $2g-3$, and $b_g$ is a polynomial in $a_2,b_2$ of degree $2g-2$, where $\deg a_2:=1$ and $\deg b_2:=2$.  
\end{itemize}
\end{rmk}

\medskip

Let us now compute the central invariant for the above conjectural $2$-parameter family of bihamiltonian structures. With the compatible flat metrics $\phi_1(v), \phi_2(v)$ given in~\eqref{metric}, the canonical coordinate $r$ is given by
$$
r= \frac{\phi_2(v)}{\phi_1(v)} = \frac{1}{32 \sqrt{\Delta}}\bigl(e^{32 \sqrt{\Delta} v}-1\bigr),
$$
with $r=v$ when $\Delta=0$. In the canonical coordinate the flat pencil reads
$$
r g(r) + \lambda g(r),
$$
where 
$$
g(r) = \bigl(1+32 \sqrt{\Delta} \, r\bigr)^{\frac32-\frac{15 a_2}{2 \sqrt{\Delta}}}.
$$
Here it is understood that when $\Delta=0$, $g(r)=e^{-240 a_2 r}$. By using~\cite[formula (1.49)]{DLZ06}, we find the central invariant
$$
c(r) = \frac1{24} \bigl(1+32\sqrt{\Delta} \, r\bigr)^{-\frac12+\frac{15 a_2}{2 \sqrt{\Delta}}}.
$$	
Here when $\Delta=0$ it is understood that $c(r)=\frac1{24} e^{240 a_2 r}$.

\medskip

We observe that $g(r), c(r)$ satisfy the relation
$$
g(r) c(r) = \frac1{24}+\frac{4\sqrt{\Delta}}3 r.
$$

\medskip

\begin{rmk}
For the mapping WK hierarchy~\cite{YZ2023}, $g(r)$ can be arbitrary and $c(r)=\frac1{24}$, under an appropriate choice of a Poisson pencil. For the bihamiltonian DR hierarchy associated to a rank-$1$ F-CohFT, the product $g(r)c(r)$, under an appropriate choice of a Poisson pencil, is an affine-linear function. 
\end{rmk}

\medskip

Let us look at four examples: 

\medskip

\underline{Example~1}: $b_2=0$, $a_2\neq 0$. In this case, we have $\Delta=225 a_2^2$, 
$$
g(r)=480 a_2  r +1, \quad c(r)=\frac1{24}.
$$
It is known from the Hodge--GUE correspondence that a particular flow in the local DR hierarchy is Miura equivalent to the Volterra Lattice equation.

\medskip

\underline{Example~2}: $b_2 = \frac{-200 a_2^2}{3}$, $a_2\neq 0$. In this case, we have $\Delta= 25 a_2^2$,  
$$
g(r)=1, \quad c(r)=\frac{20 a_2}{3} r + \frac1{24}.
$$
Assuming Conjecture~\ref{bihamDRconj} is true, one can find a particular flow in the corresponding local DR hierarchy that is Miura equivalent to the CH equation, noticing that the latter is bihamiltonian and has the above $g(r), c(r)$ (after appropriate rescalings and choice of a Poisson pencil, see Remark~\ref{remark:degree of ag and bg} and formula~\eqref{eq:change of g and c}). Let us give the Miura equivalence more directly at the approximation up to $O(\epsilon^8)$. Taking into account Remark~\ref{remark:degree of ag and bg}, we can assume $a_2=\frac1{20}$ for simplicity. Consider 
$$
f(u) = e^{8 u}
$$
in~\eqref{localDRh}, and consider the Miura-type transformation 
$$
R =e^{32\sqrt{\Delta} u} = e^{8 u}.
$$
Then we have 
\begin{align}
\frac{\partial R}{\partial t_f}& = R R_1 + \e^2\biggl(
\frac{2 R}{3}R_3+\frac{1}{3 R}R_1^3-\frac{2}{3}R_1 R_2 \biggr)\nn\\
&+\e^4\biggl(\frac{2 R R_5}{3}+\frac{47 R_1^5}{10 R^3}-\frac{602 R_2 R_1^3}{45 R^2}+\frac{44 R_3 R_1^2}{9 R}+\frac{149 R_2^2 R_1}{18 R}-\frac{4 R_4 R_1}{3}-\frac{11 R_2 R_3}{3}\biggr) \nn\\
& + \e^6\biggl(
\frac{78137 R_1^7}{567 R^5}-\frac{54562 R_2 R_1^5}{105 R^4}+\frac{24752 R_3 R_1^4}{135 R^3}+\frac{77426 R_2^2 R_1^3}{135 R^3}-\frac{15482 R_4 R_1^3}{315 R^2}\nn\\
&-\frac{10972 R_2 R_3 R_1^2}{35 R^2}-\frac{155992 R_2^3 R_1}{945 R^2}+\frac{97 R_5 R_1^2}{9 R}+\frac{719 R_3^2 R_1}{18 R}+\frac{5359 R_2 R_4 R_1}{90 R} \nn\\
&-2 R_6 R_1+\frac{20354 R_2^2 R_3}{315 R}+\frac{2 R R_7}{3}-\frac{139 R_2 R_5}{18}-\frac{1177 R_3 R_4}{90}
\biggr) + O(\e^8). \nn
\end{align}
Further performing the Miura-type transformation
\begin{align}
v = & R + \epsilon^2\Bigl(\frac32 R_2 -\frac{1}{3R} R_1^2\Bigr) 
+\e^4\Bigl(\frac{15 R_4}{8}-\frac{191 R_1^4}{90 R^3}+\frac{431 R_2 R_1^2}{90 R^2}-\frac{11 R_3 R_1}{6 R}-\frac{23 R_2^2}{15 R}\Bigr)\nn\\
&+\e^6\Bigl(\frac{35 R_6}{16}-\frac{136748 R_1^6}{2835 R^5}+\frac{595381 R_2 R_1^4}{3780 R^4}-\frac{9161 R_3 R_1^3}{162 R^3}-\frac{100675 R_2^2 R_1^2}{756 R^3}+\frac{38327 R_4 R_1^2}{2520 R^2}\nn\\
&+\frac{1895 R_2 R_3 R_1}{28 R^2}+\frac{161 R_2^3}{9 R^2}-\frac{10 R_5 R_1}{3 R}-\frac{737 R_2 R_4}{70 R}-\frac{1571 R_3^2}{210 R}\Bigr)+O(\e^8),\nn
\end{align}
we verified at the approximation up to~$O(\epsilon^8)$ that it becomes the CH equation~\eqref{CHequationform2}. Let us write our conjectural statement for this particular example more explicitly. 
\begin{conj}\label{CH-DR-conj}
For the case when $b_2 = \frac{-200 a_2^2}{3}\neq 0$, there exists a unique choice of $a_3, b_3, a_4, b_4, \dots$ such that a particular flow of the corresponding local DR hierarchy after a Miura-type transformation 
(and appropriate rescalings, see Remark~\ref{remark:degree of ag and bg}) coincides with the CH equation~\eqref{CHequationform2}.
\end{conj}

\medskip

\begin{rmk}
It is well known that, under a certain nonlinear reciprocal transformation (aka hodograph link),  
the CH equation is transformed to the negative KdV equation. However, we do not know if this 
fact could play a role in proving Conjecture~\ref{CH-DR-conj}.
\end{rmk}

\medskip

\underline{Example~3}: $b_2\neq0$, $a_2=0$. In this case, we have $\Delta=3 b_2$,  
$$
g(r) = (1+ 32\sqrt{3b_2} \, r)^{\frac32},\quad 
c(r) = \frac1{24} (1+32\sqrt{3b_2} \, r)^{-\frac12}.
$$

\medskip

\underline{Example~4}: $\Delta=0$, or equivalently $b_2=-75 a_2^2$. In this case, we have
$$
g(r)=e^{-240a_2 r},\quad c(r)=\frac1{24} e^{240a_2 r}.
$$

\medskip

Figure~\ref{parabola} gives a summary of relations of $a_2,b_2$ in the above four examples. 
\begin{figure}[htbp]
\centering
\includegraphics[scale=0.8]{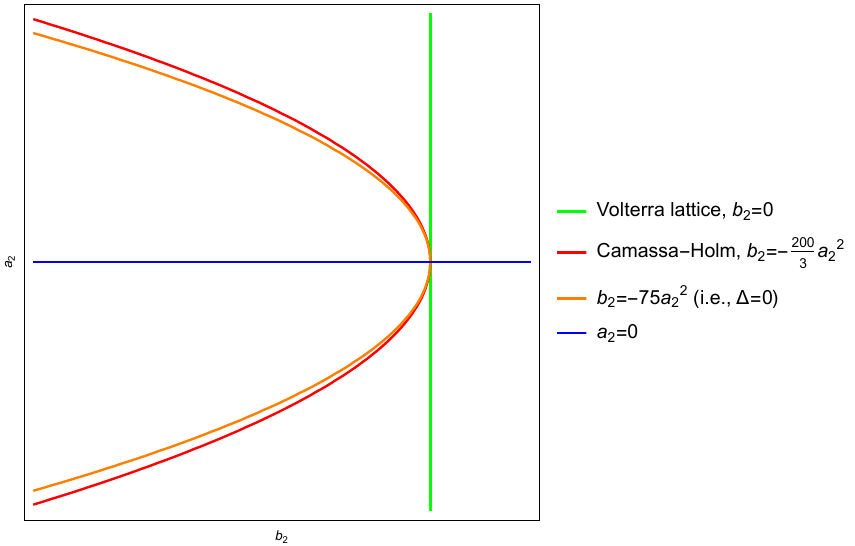}
\caption{Relations in the $2$-parameter family $(a_2,b_2)$ for four examples.}
\label{parabola}
\end{figure}

\medskip

\section{More on Miura invariants of the DR hierarchy}\label{Miura-invariants-transf}

Consider the DR hierarchy corresponding to an F-CohFT $\{c_{g,n+1}\}$. We know that the Miura invariants of the DR hierarchy are constants. Our goal is to give a systematic study of the fact that the local DR hierarchy after a Miura-type transformation could coincide with an integrable perturbation of the local RH hierarchy with the ALM normal form having not necessarily constant Miura invariants.

\medskip

Consider arbitrary functions $f(u),h(u)\in\cS(u)$. Since $\int g_f dx$ is a conserved quantity of the flow $\frac{\d}{\d t_h}$, we have $\frac{\d g_f}{\d t_h}=\d_x P_{f,h}$ for some differential polynomial $P_{f,h}$. Suppose $\frac{\d f}{\d u}\ne 0$. Then $u\mapsto v=g_f(u,u_1,u_2,\ldots,\eps)$ 
is a Miura-type transformation, and one can immediately see that the flows $\frac{\d}{\d t_{h_d}}$, where $h_d(u):=\frac{f(u)^d}{d!}$, $d\ge 0$, have the following form in the coordinate $v$:
\begin{gather}\label{eq:Miura transformation of DR hierarchy}
\frac{\d v}{\d t_{h_d}}=\d_x\lb\frac{v^{d+1}}{(d+1)!}+O(\eps)\rb,\quad h_d(u)=\frac{f(u)^d}{d!},\quad d\ge 0.
\end{gather}
Therefore, we again obtain an integrable perturbation of the RH hierarchy of the form~\eqref{conservelawhierarchy}. The next theorem describes the Miura invariants of the resulting perturbation of the RH hierarchy.

\medskip

\begin{thm}\label{theorem:new Miura invariants}
Consider an arbitrary F-CohFT $\{c_{g,n+1}\}$, a function $f(u)\in\cS(u)$ satisfying $\frac{\d f}{\d u}\ne 0$, and the Miura-type transformation $u\mapsto v=g_f(u,u_1,u_2,\ldots,\eps)$. Let us denote by $c_2=\frac{1}{12}, c_4,c_6,\ldots$ the Miura invariants of the DR hierarchy and by $C_{2g}(v)$ the Miura invariants of the perturbation~\eqref{eq:Miura transformation of DR hierarchy}. Then $C_{2g}(v)=b_{2g}(f^{-1}(v))$, where the functions $b_{2g}(u)\in\cS(u)$ are determined by the recursive relation
\begin{gather}\label{eq:recursion for b}
b_{2g}(u)=c_{2g} f'(u)+\sum_{k=1}^{g-1}\frac{2k+1}{2g+1}c_{2k} b_{2g-2k}'(u),\quad g\ge 1,
\end{gather}
or, alternatively, are explicitly given by
\begin{gather}\label{eq:explicit formula for b}
b_{2g}(u)=\sum_{k=1}^g\sum_{\substack{g_1,\ldots,g_k\ge 1\\g_1+\ldots+g_k=g}}\frac{2g_1+1}{2g+1}\frac{2g_2+1}{2(g-g_1)+1}\cdots\frac{2g_{k-1}+1}{2(g_{k-1}+g_k)+1}f^{(k)}(u)\prod_{i=1}^k c_{2g_i},
\end{gather}
or by
\beq\label{genbc}
\p_z (z c(z)) = \frac{\p_z (z b(u,z))}{\p_u (f(u)+b(u,z))},
\eeq
where $c(z):=\sum_{g\ge 1}c_{2g}z^{2g}$, $b(u,z):=\sum_{g\ge1} b_{2g}(u)z^{2g}$.
\end{thm}
\begin{proof}
We have
\begin{align*}
&\frac{\p u}{\p t_1}=\d_x\biggl(\frac{u^2}{2}+\sum_{g\geq 1}\e^{2g} \sum_{\lambda\in\mathcal{P}_{2g}\atop m_1(\lambda)=0} c_{\lambda}u_{\lambda}\biggr),&& c_\lambda\in\mbC,\\
&\frac{\p u}{\p t_f}=\p_x
\biggl(F(u)+ \sum_{g\geq 1}\e^{2g}\sum_{\lambda\in\mathcal{P}_{2g}}b_{\lambda}(u)u_{\lambda}\biggr),&& b_\lambda(u)\in\cS(u),
\end{align*}
where $F'(u)=f(u)$. Then $C_{2g}(v)=b_{2g}(f^{-1}(v))$. We know~\cite{LZ06} that the relation $\p_{t_f}\p_{t_1}(u)-\p_{t_1} \p_{t_f}(u)=0$ uniquely determines all the functions $b_\lambda(u)$ in terms of the coefficients $c_\lambda$ and the function $f(u)$. Considering $c_\lambda$ as arbitrary constants and $f(u)$ as an arbitrary function, we have $\p_{t_f}\p_{t_1}(u)-\p_{t_1} \p_{t_f}(u)=\d_x P$ for some differential polynomial $P$. One can easily check that the coefficient of $\e^{2g}u_{2g+1}$, $g\ge 0$, in $P$ is equal to zero, and therefore this does not give any constraints for $c_\lambda$ and $f$. The coefficient of~$\e^{2g}u_1 u_{2g}$ in $P$ is nontrivial, and equating it to zero gives the recursive relation~\eqref{eq:recursion for b}. The explicit formula~\eqref{eq:explicit formula for b} can be easily derived from it using the induction on~$g$. Formula~\eqref{genbc} follows from~\eqref{eq:recursion for b}.
\end{proof}

\medskip

\begin{cor}\label{corlinear}
Under the assumptions of Theorem~\ref{theorem:new Miura invariants}, let us take $f(u)=e^{\alpha u}$, $\alpha\in\mbC^*$. Then $C_{2g}(v)$ are all linear functions, $C_{2g}(v)=\beta_{2g}v$, where the coefficients $\beta_{2g}\in\mbC$ are determined by the recursion
\begin{gather}\label{eq:recursion for beta}
\beta_{2g}=\alpha\biggl( c_{2g}+\sum_{k=1}^{g-1}\frac{2k+1}{2g+1}c_{2k} \beta_{2g-2k}\biggr),\quad g\ge 1.
\end{gather}
Moreover, the generating series $c(z)$ and $\beta(z):=\sum_{g\ge 1}\beta_{2g}z^{2g}$ are related by
\begin{gather}\label{eq:c(z) and beta(z)}
\d_z(z c(z))=\frac{\d_z(z\beta(z))}{\alpha(1+\beta(z))}.
\end{gather}
\end{cor}
\begin{proof}
One can immediately see that for $f=e^{\alpha u}$ the solution of the recursion~\eqref{eq:recursion for b} has the form $b_{2g}(u)=\beta_{2g}e^{\alpha u}$, where $\beta_{2g}$ is determined by the recursion~\eqref{eq:recursion for beta}. Then $C_{2g}(v)=b_{2g}(f^{-1}(v))=\beta_{2g} v$. Relation~\eqref{eq:c(z) and beta(z)} follows directly from~\eqref{eq:recursion for beta} or~\eqref{genbc}.
\end{proof}

\medskip

\begin{exa}
Under the assumptions of Corollary~\ref{corlinear}, the first few $C_{2g}(v)$ are
$$
C_2(v)=\frac{\alpha}{12} v,\quad 
C_4(v)=\Bigl(\thin\frac{\alpha^2}{240}+\alpha  c_4\Bigr) v, \quad
C_6(v)=\Bigl(\thin\frac{\alpha^3}{6720}+\frac{2 \alpha ^2 c_4}{21}+\alpha  c_6\Bigr) v.
$$ 
\end{exa}

\medskip

\begin{rmk}
We know that for arbitrary constants $c_4,c_6,\ldots$ there exists a unique formal power series $R(z)\in 1+z\mbC[[z]]$ such that the Miura invariants of the DR hierarchy corresponding to the F-CohFT $\{R(z).c_{g,n+1}^\triv\}$ are equal to $c_2=\frac{1}{12},c_4,c_6,\ldots$. Together with Corollary~\ref{corlinear} this implies that for arbitrary linear functions $C_2(v),C_4(v),\ldots$ satisfying $C_2(v)\ne 0$ there exist unique $R(z)$ and $\alpha\ne 0$ such that the Miura invariants of the hierarchy~\eqref{eq:Miura transformation of DR hierarchy} with $f(u)=e^{\alpha u}$ are equal to $C_2(v),C_4(v), \dots$. We conclude that, assuming Conjecture~\ref{conjecture:ALM} is true, the first flow of 
any integrable perturbation of the RH hierarchy in the ALM normal form~\eqref{ALMnormalform} with $C_{2g}(v)$ being linear functions and $C_2(v)\ne 0$ is Miura equivalent to a particular flow of the local DR hierarchy of a specific F-CohFT.
\end{rmk}

\medskip

\begin{exa}
For $C_{2g}(v)=\gamma v$ with $\gamma\ne 0$, corresponding to $\beta(z)=\gamma\frac{z^2}{1-z^2}$, we compute
\begin{align}
\frac{\d_z(z\beta(z))}{\alpha(1+\beta(z))}=\frac{z^2}{\alpha}\frac{3\gamma-\gamma z^2}{(1-z^2)(1-(1-\gamma)z^2)}=&\frac{z^2}{\alpha}\lb\frac{2}{1-z^2}+\frac{3\gamma-2}{1-(1-\gamma)z^2}\rb \nn\\
=&\sum_{g\ge 1}z^{2g}\lb\frac{2}{\alpha}+\frac{3\gamma-2}{\alpha}(1-\gamma)^{g-1}\rb, \nn
\end{align}
which gives $c(z)=\sum_{g\ge 1}\frac{z^{2g}}{2g+1}\lb\frac{2}{\alpha}+\frac{3\gamma-2}{\alpha}(1-\gamma)^{g-1}\rb$. Requiring $c_2=\frac{1}{12}$ gives $\alpha=12\gamma$ and finally
$$
c(z)=\sum_{g\ge 1}\frac{z^{2g}}{2g+1}\lb\frac{1}{6\gamma}+\frac{3\gamma-2}{12\gamma}(1-\gamma)^{g-1}\rb.
$$
\begin{itemize}
\item For $\gamma=\frac{2}{3}$, which gives the Miura invariants of the CH hierarchy, we obtain $\alpha=8$ and $c_{2g}=\frac{1}{4(2g+1)}$.

\smallskip

\item For $\gamma=\frac{3}{4}$, which gives the Miura invariants of the DP hierarchy, we obtain $\alpha=9$ and $c_{2g}=\frac{1}{2g+1}\lb\frac{2}{9}+\frac{1}{36\cdot 4^{g-1}}\rb$.
\end{itemize}
\end{exa}

\bigskip

\noindent {\bf Note added.}  
Consider the bihamiltonian test for the equation
\beq\label{ALMnormalform-constants}
\frac{\p w}{\p t}=\frac{\p \omega}{\p x},\quad 
\omega=\frac{w^2}{2}
+\sum_{k\geq 2}\e^{k} \sum_{\lambda\in\mathcal{P}_{k}\atop m_1(\lambda)=0} c_{\lambda}w_{\lambda},
\eeq
where $c_{\lambda}$ are arbitrary constants and $c_{2}=\frac1{12}$.
We can prove the following
\begin{thm}\label{thmadd}
For any pair of constants $(c_4,c_{2^2})=:(a_2,b_2)$, there exists a unique choice of constants $\{c_{\lambda}\}$ 
such that~\eqref{ALMnormalform-constants} is bihamiltonian.
\end{thm}
\begin{proof}
{\it Existence}. Start with a Poisson pencil of hydrodynamic type $\cP_2^{[0]}+\lambda \cP_1^{[0]}$ with
$\cP_{1}^{[0]}=g(r)\p_x +\frac12 g'(r)r_x$ and $\cP_{2}^{[0]}=r g(r)\p_x +\frac12(g(r)+r g'(r))r_x$.
For $(g(r),c(r))=(r^{\alpha}, C r^{1-\alpha})$ with $\alpha,C$ being arbitrary constants 
(the relations with $(a_2,b_2)$ can be deduced from Section~\ref{bihamitest}), 
we first show that there exist a Poisson pencil $\cP_2+\lambda \cP_1$
with $\cP_1=\cP_{1}^{[0]}$ and $\cP_2=\sum_{g\ge0}\e^{2g}\cP_{2}^{[2g]}$, 
such that it has central invariant $c(r)$ and 
${\rm gr} \, \cP_{2}^{[2g]}=(\alpha+1) \cP_{2}^{[2g]}$, $g\ge1$. Here ${\rm gr}:=\sum_{k\ge0} r_k \frac{\p}{\p r_k}$.
Indeed, the Poisson pencil condition simplifies to
\beq
\Bigl[\cP_1^{[0]}, \cP_2^{[2g]}\Bigr] = 0, \quad \cP_2^{[2g]} = -\frac12 \sum_{j=1}^{g-1} \Bigl[\cP_2^{[2j]},\cP_2^{[2g-2j]}\Bigr],\quad g\ge1,
\eeq
where $[\,,]$ is Schouten bracket. For $g=1$, using the explicit formula given in~\cite[Theorem 7]{CPS16},
we find existence of $\cP_{2}^{[2]}$ such that ${\rm gr} \, \cP_{2}^{[2]}=(\alpha+1) \cP_{2}^{[2]}$. For $g\ge2$, 
by vanishing of the bihamiltonian cohomologies $BH^3_d$ for $d\ge4$~\cite{CPS16, LZ13},
we can prove inductively 
existence of $\cP_{2}^{[2g]}$ with ${\rm gr} \, \cP_{2}^{[2g]}=(\alpha+1) \cP_{2}^{[2g]}$.
Next, consider a specific flow of the bihamiltonian hierarchy given by the bihamiltonian operators $ \cP_1, \cP_2$~\cite{DLZ06}:
\beq
\frac{\p r}{\p t} = \log(r) \, \frac{\p r}{\p x} + O(\e^2).
\eeq
Let $w=w(r,r_1,r_2,\dots;\e)=\log(r)+O(\e^2)\in\mathcal{A}_r[[\e^2]]$ be a hamiltonian density of the bihamiltonian hierarchy. 
Then we have ${\rm gr} \, w = 1$, which implies ${\rm gr} = \frac{\p}{\p w}$. 
Performing the Miura-type transformation $r\to  w$, we have
$\frac{\p w}{\p t} =  w  \frac{\p w}{\p x} + O(\e^2)$,
which has the form~\eqref{ALMnormalform-constants}. (This flow is actually even in~$\e$.) 
For $(g(r),c(r))=(e^{\beta r}, \frac{1}{24} e^{-\beta r})$ with $\beta$ being arbitrary constant, the proof is similar.

 {\it Uniqueness}. Performing the bihamiltonian test, we can verify the statement up to and including~$\e^{6}$. 
The uniqueness is then proved by using the argument similar to~\cite{YZ2023} 
(cf.~Corollary~\ref{corollary:from P1P2 to Q}).
\end{proof}

We note that the ALM conjecture together with Theorem~\ref{thmadd} implies Conjecture~\ref{bihamDRconj}. 
Conjecture~\ref{bihamDRconj} leads to Conjecture~\ref{CH-DR-conj}. Last but not least, Theorem~\ref{thmadd}
implies that a particular flow of the local CH hierarchy is Miura equivalent to an equation of the form~\eqref{ALMnormalform-constants}.
These give strong evidences for Conjectures \ref{bihamDRconj}, \ref{CH-DR-conj}.

\end{document}